\documentclass{llncs}

\usepackage{amssymb,amsmath}
\usepackage{graphicx}
\usepackage{xcolor}

\usepackage{algorithmic}
\usepackage{algorithm}
\floatname{algorithm}{Algorithm}

\newcommand{\tj}{\textsc{TJ-reconfiguration}}

\title{Token Jumping in minor-closed classes}
\author{Nicolas Bousquet \inst{1} \thanks{The author is partially supported by ANR project STINT (ANR-13-BS02-0007).}\and
Arnaud Mary \inst{2} \thanks{The author is partially supported by ANR project GraphEn (ANR-15-CE40-0009).}
\and Aline Parreau \inst{3}\thanks{The author is partially supported by ANR project GAG (ANR-14-CE25-0006).}
}
\institute{Laboratoire G-SCOP, CNRS, Univ. Grenoble Alpes, Grenoble, France.
. \and Univ Lyon, Universit\'e Lyon 1, LBBE CNRS UMR 5558, F-69622 Lyon, France.
 \and Univ Lyon, Universit\'e Lyon 1, LIRIS UMR CNRS 5205, F-69621, Lyon, France.
}
\begin{document}

\maketitle

\begin{abstract}
Given two $k$-independent sets $I$ and $J$ of a graph $G$, one can ask if it is possible to transform the one into the other in such a way that, at any step, we replace one vertex of the current independent set by another while keeping the property of being independent. Deciding this problem, known as the Token Jumping (TJ) reconfiguration problem, is PSPACE-complete even on planar graphs. Ito et al. proved in 2014 that the problem is FPT parameterized by $k$ if the input graph is $K_{3,\ell}$-free.

We prove that the result of Ito et al. can be extended to any $K_{\ell,\ell}$-free graphs. In other words, if $G$ is a $K_{\ell,\ell}$-free graph, then it is possible to decide in FPT-time if $I$ can be transformed into $J$. As a by product, the TJ-reconfiguration problem is FPT in many well-known classes of graphs such as any minor-free class.
\end{abstract}

\section{Introduction}

Reconfiguration problems arise when, given an instance of a problem and a solution to it, we make elementary changes to transform the current solution into another. The objective can be to sample a solution at random, to generate all possible solutions, or to reach a certain desired solution.
Many types of reconfiguration problems have been introduced and studied in various fields. For instance reconfiguration of graph colorings~\cite{BonamyB13,Feghali0P15}, Kempe chains~\cite{BonamyBFJ15,FeghaliJP15}, shortest paths~\cite{Bonsma13}, satisfiability problems~\cite{Gopalan09} or dominating sets~\cite{LokshtanovMPRS15} have been studied. For a survey on reconfiguration problems, the reader is referred to~\cite{vHeuvel13}.  Our reference problem is the independent set problem.

In the whole paper, $G=(V,E)$ is a graph where $n$ denotes the size of $V$, and $k$ is an integer. For standard definitions and notations on graphs, we refer the reader to~\cite{Diestel2005}. A \emph{$k$-independent set} of $G$ is a subset of vertices of size $k$ such that no two elements of $S$ are adjacent. The $k$-independent set reconfiguration graph is a graph where vertices are $k$-independent sets and two independent sets are adjacent if they are ``close'' to each other.

Three possible definitions of adjacency between independent sets have been introduced.
In the \emph{Token Addition Removal} (TAR) model~\cite{BonamyB14a,MouawadNRSS13}, two independent sets $I, J$ are adjacent if they differ on exactly one vertex (i.e. if there exists a vertex $u$ such that $I = J \cup \{u\}$ or the other way round). In the \emph{Token Sliding} (TS) model~\cite{BonamyB16,DemaineDFHIOOUY14,HearnD05}, vertices are moved along edges of the graph.  In the \emph{Token Jumping} (TJ) model~\cite{BonsmaKW14,Ito2011,ItoKOSUY2014,KaminskiMM12}, two $k$-independent sets $I,J$ are adjacent if the one can be obtained from the other by replacing a vertex with another one. In other words there exist $u \in I$ and $v \in J$ such that $I=(J \setminus \{v \}) \cup \{ u \}$.
In this paper, we concentrate on the Token Jumping model.

 The \emph{$k$-TJ-reconfiguration graph of $G$}, denoted $TJ_k(G)$, is the graph whose vertices are all $k$-independent sets of $G$ (of size exactly $k$), with the adjacency defined above. The \tj{} problem is defined as follows:
 \begin{quote}
  \textsc{Token Jumping (TJ)-reconfiguration} \\
  \textbf{Input:} A graph $G$, an integer $k$, two $k$-independent sets $I$ and $J$. \\
  \textbf{Output:} YES if and only if $I$ and $J$ are in the same connected component of $TJ_k(G)$.
 \end{quote}

The \tj{} problem is PSPACE-complete even for planar graphs with maximum degree $3$~\cite{HearnD05}, for perfect graphs~\cite{KaminskiMM12}, and for graphs of bounded bandwidth~\cite{Wrochna14}. On the positive side, Bonsma et al.~\cite{BonsmaKW14} proved that it can be decided in polynomial time in claw-free graphs. Kaminski et al.~\cite{KaminskiMM12} gave a linear-time algorithm on even-hole-free graphs.

\paragraph{Parameterized algorithm.}
A problem $\Pi$ is \emph{FPT} parameterized by a parameter $k$ if there exists a function $f$ and a polynomial $P$ such that for any instance $\mathcal{I}$ of $\Pi$ of size $n$ and of parameter $k$, the problem can be decided in $f(k) \cdot P(n)$.
A problem $\Pi$ \emph{admits a kernel} parameterized by $k$ (for a function $f$) if for any instance $I$ of size $n$ and parameter $k$, one can find in polynomial time, an instance $I'$ of size $f(k)$ such that $I'$ is positive if and only if $I$ is positive. A folklore result ensures that the existence of a kernel is equivalent to the existence of an FPT algorithm, but the function $f$ might be exponential. A kernel is \emph{polynomial} if $f$ is a polynomial function.

Ito et al.~\cite{ItoKOSUY2014} proved that the TJ-reconfiguration problem is W[1]-hard\footnote{Under standard algorithmic assumptions, W[1]-hard problems do not admit FPT algorithms.} parameterized by $k$. On the positive side they show that the problem becomes FPT parameterized by both $k$ and the maximum degree of $G$. Mouawad et al.~\cite{MouawadNRW14} proved that the problem is W[1]-hard parameterized by the treewidth of the graph but is FPT parameterized by the length of the sequence plus the treewidth of the graph.
In~\cite{ItoKO2014}, the authors showed that the \tj{} problem is FPT on planar graphs parameterized by $k$. They actually remarked that their proof can be extended to $K_{3,\ell}$-free graphs, i.e. graphs that do not contain any copy of $K_{3,\ell}$ as a subgraph.
In this paper (Sections \ref{sec:bounds} and \ref{sec:algo}), we prove that the result of~\cite{ItoKO2014} can be extended to any $K_{\ell,\ell}$-free graphs. More formally we show the following:

\begin{theorem}\label{thm:fpt}
 \tj{} is FPT parameterized by $k$ in $K_{\ell,\ell}$-free graphs. Moreover there exists a function $h$ such that \tj{} admits a polynomial kernel of size $\mathcal{O}(h(\ell) \cdot k^{\ell3^\ell})$ if $\ell$ is a fixed constant.
\end{theorem}

As a consequence, Theorem~\ref{thm:fpt} ensures that \tj{} admits a polynomial kernel on many classical graph classes such as bounded degree graphs, bounded treewidth graphs, graphs of bounded genus or $\mathcal{H}$-(topological) minor free graphs where $\mathcal{H}$ is a finite collection of graphs.

The proof of~\cite{ItoKO2014} consists in partitioning the graph into classes according to its neighborhood in $I \cup J$ (two vertices lie in the same class if they have the same neighborhood in $I \cup J$). The authors showed that (i) some classes have bounded size (namely those with at least $3$ neighbors in $I \cup J$); (ii) if some classes are large enough, one can immediately conclude (namely those with at most one neighbor in $I \cup J$); (iii) we can ``reduce'' classes with two neighbors in $I \cup J$ if they are too large.
As they observed, this proof cannot be directly extended to $K_{\ell,\ell}$-free graphs for $\ell \geq 4$. In this paper, we develop new tools to ``reduce'' classes. Namely, we iteratively apply a lemma of K\"{o}v\'ari, S\'os, Tur\'an \cite{KovariST54} to find a subset $X$ of vertices such that $X$ has size at most $f(k,\ell)$, contains $I \cup J$ and is such that for every $Y \subset X$, if the set of vertices with neighborhood $Y$ in $X$ is too large, then it can be replaced by an independent set of size $k$.

Note finally that the \tj{} problem is W[1]-hard parameterized only by $\ell$ since graphs of treewidth at most $\ell$ are $K_{\ell+1,\ell+1}$-free graphs. And the \tj{} problem is W[1]-hard parameterized by the treewidth~\cite{MouawadNRW14}. We left the existence of an FPT algorithm parameterized by $k+\ell$ as an open question.

\paragraph{Hardness for graphs of bounded VC-dimension.}
A natural way of extending our result would consist in proving it for graphs of bounded VC-dimension. The VC-dimension is a classical way of defining the complexity of a hypergraph that received considerable attention in various fields, from learning to discrete geometry. Informally, the VC-dimension is the maximum size of a set on which the hyperedges of the hypergraph intersect on all possible ways. A formal definition will be provided in Section~\ref{sec:vcdim}. In this paper, we define the VC-dimension of a graph as the VC-dimension of its closed neighborhood hypergraph, which is the most classical definition used in the literature (see~\cite{BousquetLLPT15} for instance).

Bounded VC-dimension graphs generalize $K_{\ell,\ell}$-free graphs since $K_{\ell,\ell}$-free graphs have VC-dimension at most $\ell + \log \ell$. One can naturally ask if our results can be extended to graphs of bounded VC-dimension. Unfortunately the answer is negative since we can obtain as simple corollaries of existing results that the \tj{} problem is NP-complete on graphs of VC-dimension $2$ and $W[1]$-hard parameterized by $k$ on graphs of VC-dimension $3$. We complete these results in Section \ref{sec:vcdim} by showing that the problem is polynomial on graphs of VC-dimension $1$. The parameterized complexity status remains open on graphs of VC-dimension~$2$.
\vspace{-5pt}

\section{Density of $K_{\ell,\ell}$-free graphs}\label{sec:bounds}
K\"{o}v\'ari, S\'os, Tur\'an~\cite{KovariST54} proved that any $K_{\ell,\ell}$-free graph has a sub-quadratic number of edges. The initial bound of~\cite{KovariST54} was later improved, see e.g.~\cite{Furedi}. 

\begin{theorem}[K\"{o}v\'ari, S\'os, Tur\'an~\cite{KovariST54}]\label{thm:kst}
  Let $G$ be a $K_{\ell,\ell}$-free graph on $n$ vertices. Then $G$ has at most $ex(n,K_{\ell,\ell})$ edges, with
\[ex(n,K_{\ell,\ell}) \leq \Big(\frac{\ell-1}{2}\Big)^{1/\ell} \cdot n^{2-1/\ell}+\frac 12 (\ell-1) n. \]
\end{theorem}

As a corollary, for every $\ell$, there exists a polynomial function $P_{\ell}$ such that every $K_{\ell,\ell}$-free graph with at least $n \geq P_{\ell}(k)$ vertices contains a stable set of size at least $k$.
Note that in the following statements, we did not make any attempt in order to optimize the functions.

\begin{corollary}\label{coro:indepk}
 Every $K_{\ell,\ell}$-free graph with $k\ell(4k)^\ell$ vertices contains an independent set of size $k$.
\end{corollary}

\begin{proof}
Let us first prove the following fact.
Every $K_{\ell,\ell}$-free graph $G'$ on at least $n(k,\ell)= \ell(4k)^\ell$ vertices has a vertex of degree less than $\frac{n}{k}$.
Assume by contradiction that every vertex has degree at least $\frac{n}{k}$. So the number of edges of the graph is at least $\frac{n^2}{2k}$. Since $n\geq\ell(4k)^\ell$, we have in particular $n>2k\ell$ and thus $\frac 12 \ell n < \frac{n^2}{4k}$. We also have $n>\frac{\ell(4k)^\ell}{2}$ and thus $(\frac{\ell}{2n})^{1/\ell}<\frac{1}{4k}$. Using Theorem~\ref{thm:kst}, we obtain the following upper bound on the number of edges:
 \begin{align*}
  |E| \leq  \Big(\frac{\ell-1}{2}\Big)^{1/\ell} \cdot n^{2-1/\ell}+\frac 12 (\ell-1) n < & \left(\frac{\ell}{2n}\right)^{1/\ell}\cdot n^{2} + \frac 12 \ell n\\
<& \frac{n^2}{4k}+\frac{n^2}{4k}= \frac{n^2}{2k}
 \end{align*}
which gives a contradiction with the lower bound on the number of edges.

To conclude, let us prove the corollary by induction on $k$. For $k=1$, the result is straightforward.
Consider a graph $G$ on at least $n'(k,\ell)=k \cdot \ell(4k)^{\ell}$ vertices with $k \geq 2$. Let $z$ be a vertex of minimum degree. Since the graph has size at least $n'(k,\ell) \geq n(k,\ell)$, there exists a vertex of degree less than $\frac{n}{k}$. Add $z$ in the independent set and delete $N(z) \cup \{ z \}$ from $G$. The remaining graph has size at least $(k-1)\frac{n}{k} \geq (k-1)\cdot \ell(4k)^\ell \geq n'(k-1,\ell)$. By induction, it has an independent set $I$ of size $(k-1)$. Then $I \cup \{ z \}$ is an independent set since all the neighbors of $z$ in $C$ have been deleted, which concludes the proof. \qed
\end{proof}

We will also need a ``bipartite version'' of both Theorem~\ref{thm:kst} and Corollary~\ref{coro:indepk}.
\begin{theorem}[K\"{o}v\'ari, S\'os, Tur\'an~\cite{KovariST54}]\label{thm:kstbip}
Let $G=((A,B),E)$ be a $K_{\ell,\ell}$-free bipartite graph where $|A|=n$ and $|B|=m$. The number of edges of $G$ is at most
\[ex(n,m,K_{\ell,\ell}) \leq (\ell-1)^{1/\ell} \cdot (n-\ell+1)\cdot m^{1-1/\ell}+ (\ell-1) m. \]
\end{theorem}

\begin{corollary}\label{coro:nblargefraction}
Let $\ell \geq 3$.
 Let $G$ be a $K_{\ell,\ell}$-free graph and $C$ be a subset of vertices of size at least $(3\ell)^{4\ell}$.
 There are at most $(3\ell)^{2\ell}$ vertices of $G$ incident to a fraction of at least $\frac{1}{8\ell}$ of the vertices of $C$.
\end{corollary}

\begin{proof}
 Assume by contradiction that at least $m(\ell)=(3\ell)^{2\ell}$ vertices of $G$ are incident to at least $\frac{1}{8\ell}|C|$ vertices of $C$. Let $X$ be a subset of size $(3\ell)^{2\ell}$ that satisfies this property. Let us consider the bipartite $(X,C \setminus X)$. Let us denote by $n$ the size of $C$.

Let us first give a lower bound on the number of edges. Since all the vertices of $X$ are incident to a fraction of at least $\frac{1}{8\ell}$ of the vertices of $C$, the bipartite graph induced by $(X,C \setminus X)$ has at least $m \cdot \frac{n}{8\ell}-m^2 \geq m \cdot (\frac{n}{8\ell}-m) \geq mn (\frac{1}{8\ell} - \frac{1}{(9\ell)^{2\ell}}) \geq \frac{1}{9\ell} mn$ edges.  The negative term is due to the fact that some vertices of $X$ might be included in $C$, and then some of the edges from $X$ to $C$ might be edges from $X$ to $X$ rather than edges from $X$ to $C \setminus X$.

Conversely, let us provide an upper bound on the number of edges to obtain a contradiction. By Theorem~\ref{thm:kstbip}, the number of edges is at most

 \begin{align*}
  (\ell-1)^{1/\ell} \cdot (n-\ell+1)\cdot m^{1-1/\ell}+ (\ell-1) m <&  \frac{(\ell-1)^{1/\ell}}{m^{1/\ell}} \cdot m n +  (\ell-1) m \\
  \leq& \frac{\ell^{1/\ell}}{m^{1/\ell}} \cdot m n +  \ell m \frac{n}{(3\ell)^{4\ell}} \\
  \leq& \Big(\frac{\ell^{1/\ell}}{(3\ell)^{2\ell/\ell}}+ \frac{\ell}{(3\ell)^{4\ell}}\Big)\cdot m n \\
  \leq& \frac{1}{9 \ell} \cdot mn
 \end{align*}

 which contradicts the lower bound on the number of edges. \qed
\end{proof}

\section{Polynomial kernel on $K_{\ell,\ell}$-free graphs}\label{sec:algo}

In this section we prove the following that implies Theorem~\ref{thm:fpt}.

\begin{theorem}\label{thm:kernel}
 The \tj{} problem admits a kernel of size $h(\ell)\cdot k^{\ell \cdot 3^\ell}$.
\end{theorem}
Let $G$ be a $K_{\ell,\ell}$-free graph and $k$ be an integer. Let $I$ and $J$ be two distinct independent sets of size $k$.
Two vertices $a$ and $b$ are \emph{similar} for a subset $X$ of vertices if both $a$ and $b$ have the same neighborhood in $X$. A \emph{similarity class} (for $X$) is a maximum subset of vertices of $V \setminus X$ with the same neighborhood in $X$.

In Section~\ref{subsec:algo}, we present basic facts and describe the kernel algorithm.
In Section~\ref{sec:sizereduced}, we bound the size of the graph returned by the algorithm. It will be almost straightforward to see that the size of this graph is a function of $k$ and $\ell$. However, we will need additional lemmas to prove that its size is at most $h(\ell) \cdot k^{\ell \cdot 3^\ell}$ and that the algorithm is polynomial when $\ell$ is a fixed constant. Section~\ref{sec:correct} is devoted to prove that the algorithm is correct.

\subsection{The algorithm}\label{subsec:algo}

Let us first briefly informally describe the behavior of Algorithm~\ref{alg1}. During the algorithm, we will update a set $X$ of important vertices. At the beginning of the algorithm we set $X=I \cup J$. At each step of the algorithm, at most $f(k,\ell)$ vertices will be added to $X$. So each similarity class of the previous step will be divided into at most $2^{f(k,\ell)}$ parts. 
The main ingredient of the proof (essentially) consists in showing that, at the end of the algorithm, either the size of a class is bounded or the whole class can be replaced by an independent set of size $k$.
As a by-product, the size of the graph can be bounded by a function of $k$ and $\ell$.

Let $X$ be a set of vertices containing $I \cup J$. The \emph{rank} of a similarity class $C$ for $X$ is the number of neighbors of $C$ in $X$.
Our proof consists in applying different arguments depending on the rank of the similarity classes. We actually consider the $3$ distinct types of classes: classes of rank at least $\ell$, classes of rank at most $1$ and classes of rank at least $2$ and at most $\ell-1$.
The size of the class $C$ in the first two cases can be bounded as shown  in~\cite{ItoKO2014} for $\ell=3$.

\begin{lemma}\label{lem:rankl}
 The size of a class $C$ of rank at least $\ell$ is at most $\ell-1$.
\end{lemma}

\begin{proof}
All the vertices of $C$ have the same neighbors in $X$. Let $Y$ be the set of neighbors of $C$ in $X$. The set $(C,Y)$ is a complete bipartite graph. Since $G$ is $K_{\ell,\ell}$-free, the size of $C$ is at most $\ell-1$ since $Y$ has size at least $\ell$. \qed
\end{proof}

\begin{lemma}\label{lem:rank1}
Let $X$ be a set of vertices containing $I \cup J$. If the size of a class $C$ of rank at most $1$ for $X$ is at least $k \ell (4k)^{\ell}$, then it is possible to transform $I$ into $J$.
\end{lemma}

\begin{proof}
By Corollary~\ref{coro:indepk}, the graph $G[C]$ admits an independent set $\{ z_1,\ldots,z_k \}$ of size $k$.
Let us denote by respectively $x_1,\ldots,x_k$ and $y_1,\ldots,y_k$ the vertices of $I$ and $J$. Moreover, since $|N(C) \cap X| \leq 1$ and that $I \cup J \subseteq X$, we can assume without loss of generality that no vertex of $C$ is incident to $I \cup J \setminus \{x_1\}$.

We can now transform $I$ into $J$. We first move the token from $x_1$ to $z_1$ and we continue by moving tokens from $x_i$ to $z_i$ for every $i \geq 2$. At each step $\{z_1,\ldots,z_i,x_{i+1},\ldots,x_k \}$ is an independent set. We can similarly prove that $Z$ can be transformed into $J$ which concludes the proof. \qed
\end{proof}

Our approach to deal with the remaining classes consists in adding vertices in $X$ to increase their ranks. Since it is simple to deal with a class of rank at least $\ell$, it provides a way to simplify classes.
However, some vertices might not be incident to the new vertices of $X$, and then their ranks do not increase. The central arguments of the proof consists in proving that we can deal with these vertices if we repeat a ``good'' operation at least $2k+1$ steps (see Lemma \ref{lem:lastcase2}).

The set $X$ is called the set of \emph{important vertices}. Initially, $X=I \cup J$.
We denote by $X_t$ the set $X$ at the beginning of step $t$ and $X_0=I \cup J$.
A class $D$ for $X_{t'}$ is \emph{inherited} from a class $C$ of $X_t$ if $t' >t$ and $D\subseteq C$. We say that $C$ is an \emph{ancestor} of $D$. Note that the rank of $C$ is at most the rank of $D$ since when we add vertices in $X$ the rank can only increase and the set $X$ is increasing during the algorithm.


A similarity class is \emph{big} for $X_t$ if its size is at least $g(k,\ell):=4 \cdot k \ell (4k)^{\ell}$. We say that we \emph{reduce a class $C$} when we replace all the vertices of the class $C$ by an independent set of size $k$ with the same neighborhood in $V \setminus C$: $N(c) \cap X$ where $c$ is any vertex of the class $C$.



\begin{algorithm}[!h]
\caption{Kernel algorithm}
\label{alg1}
\begin{algorithmic}[PERF]
\STATE Let $X_0= I \cup J$ \hfill {\scriptsize Initially important vertices are $I \cup J$}
 \IF{a class of rank $1$ is big}
   \STATE Return a YES instance \hfill {\scriptsize Valid operation, see Lemma~\ref{lem:rank1}}
 \ENDIF
\FOR{every $j=2$ to $\ell-1$}
\FOR{$s=0$ to $2k$}
\STATE $\mathcal{C}_t:=$ big similarity classes of rank $j$ for $X_t$. \hfill {\scriptsize We treat big classes of the current rank $j$}
\STATE $Z= \varnothing$
\FOR{every $C \in \mathcal{C}_t$}
\STATE $Y:= \{ y \in  V \setminus X_t$ such that $|N(y) \cap C| \geq \frac{|C|}{8\ell} \}$. \hfill {\scriptsize $|Y|$ is bounded, Corollary~\ref{coro:nblargefraction}}
 \IF{{$|(N(Y)\cup Y) \cap C| \leq |C|/2$}}
   \STATE Reduce the class $C$. \hfill {\scriptsize Valid operation by Lemma~\ref{clm:lastcase1}}
 \ELSE
   \STATE $Z=Z \cup Y$
 \ENDIF
 \ENDFOR
 \STATE $X_{t+1}:= X_t \cup Z$. \hfill {\scriptsize Update the important vertices (and the classes)}
\ENDFOR
\STATE Reduce all the big classes of rank $j$.  \hfill {\scriptsize Valid operation by Lemma~\ref{lem:lastcase2}}
\ENDFOR
\STATE Return the reduced graph.
\end{algorithmic}
\end{algorithm}

The classes that are in $\mathcal{C}_t$ are said to be \emph{treated} at round $t$. For these classes, we add in $X_t$ all the vertices that are incident to a $1/8\ell$-fraction of the vertices of the class. When we say that we \emph{refine} the classes at step $t$, it means that we partition the vertices of the classes according to their respective neighborhood in the new set $X_{t+1}$. \smallskip

\subsection{Size of the reduced graph}\label{sec:sizereduced}
This part is devoted to prove that the size of the graph output by the algorithm is $h(\ell) \cdot k^{\ell \cdot 3^\ell}$.
When we have finished to treat classes of rank $j$, $j<\ell$, (i.e. when the index of the first loop is at least $j+1$) then either all the classes of rank $j$ have size less than $g(k,\ell)$ or they are replaced by an independent set of size $k$. Since classes of rank $j$ cannot be created further in the algorithm, any class of rank at most $j$ has size at least $g(k,\ell)$ at the end of the algorithm. Moreover, any class of rank $\ell$ has size at most $\ell-1$ by Lemma \ref{lem:rankl}. 

A \emph{step} of the algorithm is an iteration of the second loop (variable $s$) in the algorithm. The value of $j$ at a given step of the algorithm is called the \emph{index} of the step and the value of $s$ is called the \emph{depth} of the step.

Note that at the last step of depth $2k$ and of index $i$, all the classes of rank $i$ are reduced. So at the end of this step, no class of rank $i$ is big anymore. Since the set $X_{t+1}$ contains $X_t$ for every $t$, the classes at step $t+1$ are subsets of classes of rank $t$. So there does not exist any big class of rank $i$ at any step further in the algorithm. In particular we have the following:

\begin{remark}\label{rem:size}
At any step of index $j$, no class of rank $i<j$ is big. Moreover, at the end of the algorithm, no class is big.
\end{remark}

The structure of the algorithm ensures that the algorithm ends. Actually, we have the following:

\begin{remark}\label{rem:numberrounds}
  The number of steps is equal to $(2k+1)\cdot (\ell-2)$.
\end{remark}

Using Corollary~\ref{coro:nblargefraction}, it is simple to prove that the  final size of $X$ is bounded by a function of $k$ and $\ell$. Since the number of classes only depends on $k$ and $\ell$ and each class has bounded size by Remark~\ref{rem:size}, the final size of the graph is bounded in terms of $k$ and $\ell$. The rest of this subsection is devoted to prove a better bound on the size of the final graph. We show that it is actually polynomial if $\ell$ is a fixed constant. The proof will be a consequence of the following lemma.

\begin{lemma}\label{lem:numbersets}
The size of $X$ at the end of the algorithm is at most $h'(\ell) \cdot k^{3^{\ell}}$.
\end{lemma}
\begin{proof}
 Let us denote by $N_j$ an upper bound on the maximum number of big classes of rank $j$ at any step of the algorithm. Let us first determine an upper bound on the number of vertices that are added in $X$ during the steps of index $j$. Let $t$ be a step of index $j$. The number of classes in $\mathcal{C}_t$ is at most $N_j$. Moreover, for each class in $\mathcal{C}_t$,  Corollary~\ref{coro:nblargefraction} ensures that at most $(3\ell)^{2\ell}$ vertices are added in $X$. So the size of $X_{t+1} \setminus X_t$ is at most  $N_j \cdot (3\ell)^{2\ell}$.
 Since there are $2k+1$ steps of index $j$, the number of vertices that are added in $X$ during all the steps of index $j$ is at most $N_j \cdot (2k+1) \cdot (3\ell)^{2\ell}$.
 Thus the set of important vertices $X$ at the end of the algorithm satisfies the following:
 \[ |X| \leq  2k + \sum_{j=2}^{\ell-1} \Big( N_j  \cdot (3\ell)^{2\ell} \Big) \cdot (2k+1). \]

 The remainder of the proof is devoted to find a bound on $N_j$ that immediately provides an upper bound on $|X|$. Let us prove by induction on $j$ that $N_j = f_j(\ell)\cdot k^{3^j}$, with $f_2(\ell)=4$ and $f_j(\ell)=\Big(f_{j-1}(\ell)\cdot (3\ell)^{2\ell}\Big)^j$ is a valid upper bound.

 Since classes are refinement of previous classes and by Remark~\ref{rem:size}, the number of big classes of rank $r$ is non increasing when we are considering steps of index $r$. As an immediate consequence, there are at most $(2k)^2$ big classes of rank $2$, which is the maximum number of classes of rank $2$ when $X=I \cup J$. So the results holds for $j=2$.

 Let $j>2$ and assume that $f_i(\ell)\cdot k^{3^i}$ is an upper bound on $N_i$ for any $2\leq i<j$. We say that a class of rank $j$ is \emph{created} at step $t$ if it is inherited from a class (at step $t$) of rank smaller than $j$. The maximum number of big classes of rank $j$ is at most the initial number of big classes of rank $j$ plus the number of big classes of rank $j$ created at any step of the algorithm.
 Let us count how many big classes of rank $j$ can be created at each step $t$.
 By Remark~\ref{rem:size}, if the index of the step is at least $j$, no new big classes of rank $j$ can be created.
So if a class of rank $j$ is created at step $t$, then the index of $t$ is $j-i$ with $i>0$.

Consider a big class $C$ of rank $j-i$ at step $t$. Let us count how many big classes of rank $j$ can be inherited from this class.
Since $C$ is big, the index $r$ of the step $t$ is at most $j-i$.
As we already noticed, the set $Z= X_{t+1} \setminus X_t$ has size at most $N_r \cdot (3\ell)^{2\ell} \leq N_{j-i} \cdot (3\ell)^{2\ell}$ since $(N_r)_r$ is an increasing sequence. Each class of rank $j$ inherited from $C$ must have $i$ neighbors in $Z$. Since there are at most $(N_{j-i} \cdot (3\ell)^{2\ell})^i$ ways of selecting $i$ vertices in $Z$, the class $C$ can lead to the creation of at most $(N_{j-i} \cdot (3\ell)^{2\ell})^i$ big classes of rank $j$.
By induction hypothesis, the number of big classes of rank $j-i$ is at most $N_{j-i}$. So at step $t$, the number of classes of rank $j$ that are created from classes of rank $j-i$ is at most  $N_{j-i} \cdot (N_{j-i} \cdot (3\ell)^{2\ell})^{i} \leq (N_{j-i} \cdot (3\ell)^{2\ell})^{i+1}$.

The total number of rounds of the algorithm is at most $(2k+1) \cdot (\ell-2)$ by Remark~\ref{rem:numberrounds}. So the number of big classes of rank $j$ that are created all along the algorithm from (big) classes of rank $j-i$ is at most $(N_{j-i} \cdot (3\ell)^{2\ell})^i \cdot (2k+1) \cdot (\ell-2)$.
And then the number of big classes of rank $j$ is at most:

 \[(2k)^j + \sum_{i=1}^{j-2} \Big(N_{j-i} \cdot (3\ell)^{2\ell}\Big)^{i+1} \cdot (2k+1) \cdot (\ell-2) . \]

Let us prove that this is bounded by $ f_j(\ell) \cdot k^{3^j}$.
Using the induction hypothesis, and since $3^{(j-i)}\cdot (i+1)$ is maximal for $i=1$, we have, for $1\leq i \leq j-2$,

 \begin{align*}
    (2k)^j + \sum_{i=1}^{j-2} \Big(N_{j-i} \cdot (3\ell)^{2\ell}\Big)^{i+1} \cdot (2k+1) \cdot (\ell-2)
   \leq & (2k)^j + \sum_{i=1}^{j-2} \Big(f_{j-i}(\ell) \cdot k^{3^{j-i}} \cdot (3\ell)^{2\ell}\Big)^{i+1} \cdot 3k\ell \\
   \leq & j\cdot \Big(f_{j-1} (\ell) \cdot (3\ell)^{2\ell} \Big)^{j-1} \cdot k^{2\cdot 3^{j-1}}\cdot 3k^j \ell\\
\leq & \Big(f_{j-1}(\ell)\cdot (3\ell)^{2\ell}\Big)^j \cdot k^{3^{j}}\\
\leq & f_j(\ell) \cdot k^{3^j}
 \end{align*}

 Hence, $N_j=f_j(\ell)\cdot k^{3^j}$ is an upper bound on the number of big classes of rank $j$.

Finally, the size of $X$ at the end of the algorithm is at most
\[  2k + \sum_{j=2}^{\ell-1} \Big( N_j  \cdot (3\ell)^{2\ell} \Big) \cdot (2k+1) = h'(\ell) \cdot k^{3^{\ell}}. \vspace{-25pt}
 \]
 \vspace{10pt}
\qed
\end{proof}

We have all the ingredients to determine the size of the graph at the end of the algorithm. Let us denote by $X$ the set of important vertices at the end of the algorithm.
For every subset of $X$ of size $\ell$, there exist at most $\ell-1$ vertices incident to them by Lemma~\ref{lem:rankl}. So the number of vertices with at least $\ell$ neighbor on $X$ is at most $(\ell-1)\cdot |X|^{\ell}$.
Moreover every class of rank $0$ or $1$ contains less than $g(k,\ell):=k \cdot \ell\cdot(4k)^{\ell}$ vertices by Lemma~\ref{lem:rank1}. And every class of rank between $2$ and $\ell-1$ has size at most $g(k,\ell)$ by Remark~\ref{rem:size}. Since there are $|X|^{\ell-1}$ classes of rank at most $\ell-1$, Lemma~\ref{lem:numbersets} ensures that the size $s$ of the graph returned by the algorithm is at most
\begin{align*}
  s \leq & (\ell-1) \cdot \Big(h'(\ell) \cdot k^{3^{\ell}}\Big)^\ell + \Big(h'(\ell)\cdot k^{3^\ell}\Big)^{\ell-1} \cdot k \cdot  \ell\cdot(4k)^{\ell} 
  \leq 
  h(\ell) \cdot k^{\ell \cdot 3^{\ell}}.
\end{align*}
So the size of the reduced graph has the size of the claimed kernel.

\paragraph{Complexity of the algorithm.} Let us now briefly discuss the complexity of the algorithm. By Lemma~\ref{lem:numbersets}, the size of $X$ is bounded by a function of $k$ and $\ell$ and is polynomial in $k$ if $\ell$ is a fixed constant.
The only possible non polynomial step of the algorithm would consist in maintaining an exponential number of classes. But Lemma~\ref{lem:rankl} ensures that the number of classes of rank at least $\ell$ is at most $\ell \cdot {\ell \choose |X|}$ which is polynomial if $\ell$ is a constant. So the total number of classes is at most $(\ell+1) \cdot |X|^\ell$ which ensures that this algorithm runs in polynomial time.

\subsection{Equivalence of transformations}\label{sec:correct}

This section is devoted to prove that Algorithm~\ref{alg1} is correct. To do it, we just have to prove that reducing classes does not modify the existence of a transformation. In other words we have to show that $I$ can be transformed into $J$ in the original graph if and only if $I$ can be transformed into $J$ in the reduced graph. In Algorithm~\ref{alg1}, there are two cases where we reduce a class. Lemmas~\ref{clm:lastcase1} and~\ref{lem:lastcase2} ensure that in both cases these reductions are correct.

\begin{lemma}\label{clm:lastcase1}
 Let $C$ be a big class of $\mathcal C_t$.
 Assume moreover that the set $Y$ of vertices of $V \setminus X$ incident to a fraction $\frac{1}{8\ell}$ of the vertices of $C$ satisfies $|N(Y) \cap C| \leq |C|/2$.
 Then there is a transformation of $I$ into $J$ in $G$ if and only if there is a transformation in the graph where $C$ is reduced.
\end{lemma}
\begin{proof}
Let $G$ be the original graph and $G'$ be the graph where the class $C$ has been replaced by an independent set of size $k$. We denote by $C'$ the independent set of size $k$ that replaces $C$ in $G'$. \smallskip

Assume that there exists a transformation from $I$ to $J$ in $G$. Let us prove that such a sequence also exists in $G'$. Either no independent set in the sequence contains a vertex of $C$, and then the sequence still exists in the graph $G'$. So we may assume that at least one independent set contains vertices of $C$. Let us denote by $I'$ the last independent set of the sequence between $I$ and $J$ such that the sequence between $I$ and $I'$ does not contain any vertex of $C$. In other words, it is possible to move a vertex of $I'$ to a vertex of $C$. Similarly let $J'$ be the first independent set such that the sequence between $J'$ and $J$ does not contain any vertex of $C$. Note that in the graph $G'$, the transformations of $I$ into $I'$ and of $J'$ into $J$ still exist since all the independent sets are in $G[V \setminus C]$ that is not modified.

Let us denote by $c$ the vertex of $C$ in the independent set after $I'$ in the sequence and $i_0$ the vertex deleted from $I'$.
No vertex of $(I' \setminus i_0) \cap X$ has a neighbor in $C$. Otherwise it would not be possible to move $i_0$ on $c$ since sets have to remain independent.
Thus in $G'$ we can move the vertex $i_0$ to any vertex of $C'$ and then move the remaining vertices of $I'$ to $C'$. These operations are possible since for every vertex $c'$ of $C'$, we have $N(c') \subseteq N(c) \cap X$ and $(I \cup \{ c\}) \setminus \{i_0\}$ is an independent set.
Free to reverse the sequence, a similar argument holds for $J'$. So there is a transformation from $I$ to $J$ in the graph $G'$.
\smallskip

Assume now that there exists a transformation from $I$ to $J$ in $G'$. As in the previous case, we can assume that an independent set of the transformation sequence contain a vertex of $C'$. Let us denote by $I'$ the last independent set such that the sequence between $I$ and $I'$ does not contain any vertex of $C'$. Similarly $J'$ is the first independent set such that the sequence between $J'$ and $J$ does not contain any vertex of $C'$.  Let us denote by $i_0$ and $j_0$ the vertices respectively deleted between $I'$ and the next independent set and added between the independent before $J'$ and $J'$.

Note that no vertex of $(I' \setminus i_0) \cap X$ has a neighbor in $C$.  Otherwise the independent set after $I'$ in the sequence would not be independent. Similarly no vertex of $(J' \setminus j_0) \cap X$ has a neighbor in $C$.
Let us partition $F=(I' \cup J') \setminus \{ i_0,j_0 \}$ into two sets $A$ and $B$. The set $A$ is the subset of vertices of  $F$ incident to a fraction of at least $\frac{1}{8\ell}$ of the vertices of $C$ in $G$. By hypothesis on $C$, $N(A) \cap C$ covers at most half of the vertices of $C$. Let $B$ be the complement of $A$ in $F$. Every vertex of $B$  is incident to a fraction of at most $\frac{1}{8\ell}$ of the vertices of $C$ in $G$. So $N(B)$ covers at most a quarter of the vertices of $C$.
 Let us denote by $D$ the set $C \setminus N(I' \cup J')$. The size of $D$ is at least one quarter of the size of $C$. Since $C$ is big, the size of $D$ is at least $k \cdot  \ell\cdot(4k)^{\ell}$.
 Theorem~\ref{coro:indepk} ensures that exists an independent set of size $I''$ at least $k$ in $D$. By construction of $D$, one can move $i_0$ to any vertex of $I''$. And then the remaining vertices of $I'$ to $I''$. Similarly, one can transform $I''$ into $J'$. So there exists a transformation from $I$ to $J$ in the graph $G$ that concludes this proof. \qed
\end{proof}

\begin{lemma}\label{lem:lastcase2}
Let $C$ be a class of rank $j$ when the index of the step equals $j$ and the depth of the step equals $2k$. Assume moreover that the size of $C$ is at least $4k\ell\cdot(4k)^{\ell}$.
Then there is a transformation of $I$ into $J$ in $G$ if and only if there is a transformation in the graph where $C$ is reduced.
\end{lemma}
\begin{proof}
Let $G$ be the original graph and $G'$ be the graph where $C$ is reduced. We will denote by $C'$ the independent set of size $k$ that replaces $C$ in $G'$.

Assume that there exists a transformation from $I$ to $J$ in $G$. A sequence also exists in $G'$. The proof works exactly as the proof of the first part of Lemma~\ref{clm:lastcase1}.

Assume now that there exists a transformation from $I$ to $J$ in $G'$. Let us prove that a transformation from $I$ to $J$ also exists on $G$. If none of the independent sets of the sequence contains a vertex of $C'$, the sequence still exists in $G$.  So we can assume that an independent set of the sequence contains a vertex of $C'$. Let us denote by $I'$ the last independent set such that the sequence between $I$ and $I'$ does not contain any vertex of $C'$ and $J'$ the first independent set such that the sequence between $J'$ and $J$ does not contain any vertex of $C'$.
Let us denote by $i_0$ and $j_0$ the vertices respectively deleted between $I'$ and the next independent set and added between the independent before $J'$ and $J'$. We denote by $I_0$ and $J_0$ the sets $I' \setminus i_0$ and $J' \setminus j_0$.
Note that no vertex of $(I' \setminus i_0) \cap X$ has a neighbor in $C$.  Otherwise the independent set after $I'$ in the sequence would not be independent. Similarly no vertex of $(J' \setminus j_0) \cap X$ has a neighbor in $C$.

Let us denote by $t_0$ the step of index $j$ and depth $0$. And let $t$ be the step of index $j$ and depth $2k$. In other words $t=t_0+2k$. Let $C_0,C_1,\ldots,C_{2k}=C$ be the ancestors of $C$ at round $t_0,\ldots,t_0+2k$. All these classes have rank $j$ and $C_{2k} \subseteq C_{2k-1} \subseteq \cdots C_0$. Since $|C_{2k}| \geq  g(k,\ell)$, the same holds for any class $C_i$. In particular, the class $C_i$ is big at step $t_{0}+i$. Since the class $C_i$ is not reduced at step $t_0+i$, the subset of vertices incident to a $\frac{1}{8\ell}$-fraction of the vertices of $C_i$ covers at least half of the vertices of $C_i$. In particular, for every $i <2k$, we have
\begin{equation}\label{eq1}
 |C_{i+1}| \leq |C_i|/2.
\end{equation}

Let $i \leq 2k$. Let us denote by $Y_i$ the set of vertices of $V \setminus X_{t_0+i}$ that are incident to at least $\frac{1}{8\ell}$ of the vertices of $C_i$. Any vertex $y$ of $Y_i$ has no neighbor in $C_h$ for $h>i$. Indeed the set $Y_i$ is added in $X_{t_0+i+1}$  at the end of step $t_0+i$. And by definition of $C_h$, the rank of $C_h$ is still $j$. Note moreover that, if $i \neq h$ then $Y_i$ and $Y_h$ are disjoint.
Since a vertex in $Y_h$ is not incident to  a $\frac{1}{8\ell}$-fraction of the vertices of $C_i$, Equation (\ref{eq1}) ensures for every $i< 2k$ and every vertex $x \notin Y_i$
\[ |N(x) \cap (C_i \setminus C_{i+1})| \leq \frac{|C_i \setminus C_{i+1}|}{4\ell} \]

Moreover, by definition of $Y_{2k}$, every vertex $x$ which is not in $Y_{2k}$ satisfies
\[ |N(x) \cap C_{2k}| \leq \frac{|C_{2k}|}{8\ell} \]

Since the sets $Y_0,\ldots,Y_{2k}$ are disjoint, there exists an index $i$ such that $I' \cup J'$ does not contain any vertex of $Y_i$. Let $C_i' = C_i \setminus C_{i+1}$ (or $C_i'=C_i$ if $i=2k$). Every vertex of $I' \cup J'$ is incident to at most $ \frac{|C_i'|}{4\ell}$ of the vertices of $C_i'$. So the complement of $N(I_0 \cup J_0)$ in $C_i'$, denoted by $C_i''$ has size at least
\begin{align*}
 |C_i''| &\geq \frac{|C_i'|}{2} \geq \frac{|C_i|}{4} \geq \frac{|C_{2k}|}{4} 
         \geq g(k,\ell) \geq k\ell\cdot(4k)^{\ell}.
\end{align*}
By Corollary~\ref{coro:indepk}, $C_i''$ contains an independent set $S$ of size $k$. Since $I'$ and $S$ are anticomplete (up to one vertex, namely $i_0$), one can transform the independent set from $I'$ into $S$. Similarly, one can transform $S$ into $J'$ which completes the proof. \qed
\end{proof}

\section{Bounded VC-dimension}\label{sec:vcdim}

Let $H=(V,E)$ be a hypergraph. A set $X$ of vertices of $H$ is \emph{shattered} if for every subset $Y$ of $X$ there exists a hyperedge $e$ such that $e \cap X = Y$. An intersection between $X$ and a hyperedge $e$ of $E$ is called a \emph{trace} (on $X$). Equivalently, a set $X$ is shattered if all its $2^{|X|}$ traces exist. The \emph{VC-dimension} of a hypergraph is the maximum size of a shattered set.

Let $G=(V,E)$ be a graph. The \emph{closed neighborhood hypergraph} of $G$ is the hypergraph with vertex set $V$ where $X \subseteq V$ is a hyperedge if and only if $X=N[v]$ for some vertex $v \in V$ (where $N[v]$ denotes the closed neighborhood of $v$). The \emph{VC-dimension} of a graph is the VC-dimension of its closed neighborhood hypergraph. The VC-dimension of a class of graphs $\mathcal{C}$ is the maximum VC-dimension of a graph of $\mathcal{C}$.

There is a correlation between VC-dimension and complete bipartite subgraphs. Namely, a $K_{\ell,\ell}$-free graph has VC-dimension at most $\mathcal{O}(\ell)$.
Since the \tj{} problem is W[1]-hard for general graphs and FPT on $K_{\ell,\ell}$-free graphs, one can naturally ask if this result can be extended to graphs of bounded VC-dimension.
Let us remark that the problem is W[1]-hard even on graphs of VC-dimension~$3$. This is a corollary of two simple facts.
First, to prove that the \tj{} problem is W[1]-hard on general graphs, Ito et al.~\cite{ItoKOSUY2014} showed that if the Independent Set problem is W[1]-hard on a class of graph $\mathcal{G}$, then the \tj{} problem is W[1]-hard on the class $\mathcal{G}'$ where graphs of $\mathcal{G}'$ consist in the disjoint union of a graph of $\mathcal{G}$ and a complete bipartite graph.
Note that the VC-dimension of a complete bipartite graph equals $1$. Moreover, if $\mathcal{G}$ is a class closed by disjoint union, then the VC-dimension of the class $\mathcal{G}'$ is equal to the VC-dimension of $\mathcal{G}$. Hence we have the following:

\begin{remark}
 If $\mathcal{C}$ is a class of graphs of VC-dimension at most $d$ closed by disjoint union, then the \tj{} problem on graphs of VC-dimension at most $d$ is at least as hard as the \textsc{Independent Set} problem on $\mathcal{C}$.
\end{remark}

 So any hardness result for \textsc{Independent Set} provides a hardness result for \tj{}.
 The \textsc{Independent Set} problem is W[1]-hard on graphs of VC-dimension at most $3$. Indeed, Marx proved in~\cite{Marx05} that the \textsc{Independent Set} problem is $W[1]$-hard on unit disk graphs, and unit disk graphs have VC-dimension at most $3$ (see for instance~\cite{BousquetLLPT15}). To complete the picture, we have to determine the complexity of the problem for $k=1$ and $k=2$.
 For graphs of VC-dimension $2$, the problem is NP-hard. Indeed the \textsc{Independent Set} problem is NP-complete on graphs of girth at least $5$~\cite{murphy} and this class has VC-dimension at most $2$ (see for instance~\cite{BousquetLLPT15}).

The remaining of this section is devoted to prove that \tj{} can be decided in polynomial time on graphs of VC-dimension at most $1$.

\begin{theorem}\label{thm:VC1}
 The \tj{} problem can be solved in polynomial time on graphs of VC-dimension at most $1$.
\end{theorem}

Let us give the three lemmas that permits to prove Theorem~\ref{thm:VC1} whose proofs are not included in this extended abstract.

\begin{lemma}\label{lem:neighbourhoods}
     Let $G$ be a graph of VC-dimension at most $1$ and let $u$ and $v$ be two vertices of $G$. Then one of the following holds:

     \begin{enumerate}
           \item The closed neighborhoods of $u$ and $v$ are disjoint.
           \item One of the closed neighborhoods is included in the other.
           \item $u$ and $v$ form a dominating pair.
       \end{enumerate}
\end{lemma}
\begin{proof}
    Since the pair $\left\{ u,v \right\}$ is not shattered, at least one of the traces from $\left\{ \emptyset, \left\{ u \right\},\left\{ v \right\}, \left\{ u,v  \right\} \right\}$ is missing. If the trace $\emptyset$ is missing, it means that for every vertex $w$, $N[w]\cap \left\{ u,v \right\}\neq \emptyset$ and then $u,v$ forms a dominating pair. If the trace $\{u\}$ is missing, then for every vertex $w$ such that $u\in N[w]$ we have $v\in N[w]$. Thus $N[u]$ is included in $N[v]$. Similarly, if the trace $\{v\}$ is missing, then $N[v]$ is included in $N[u]$. Finally if the trace $\left\{ u,v \right\}$ is missing then no vertex $w$ has both $u$ and $v$ in its neighborhood and then the neighborhood of $u$ and $v$ are disjoint. \qed
\end{proof}

The following lemma ensures that if the graph contains a vertex satisfying the second point of Lemma  \ref{lem:neighbourhoods}, then it can be deleted.

\begin{lemma}\label{lem:NoInclusion}
 Let $G$ be a graph of VC-dimension at most 1 and let $u$ and $v$ be two vertices such that $N[u]\subseteq N[v]$. Let $I,J$ be two independent set that do not contain $v$. Then there exists a TJ-transformation from $I$ to $J$ in $G$ if and only if there exists a TJ-transformation from $I$ to $J$ in $G':=G \setminus \{ v \}$.
\end{lemma}

\begin{proof}
 \noindent$(\Leftarrow)$ Any transformation in $G'$ also is a transformation in $G$ since $G$ contains $G'$. \smallskip

 \noindent$(\Rightarrow)$ Consider a transformation from $I$ to $J$ in $G$. Either no independent set contain $v$ and the same transformation exists in $G'$. Or some independent sets contain $v$. Now replace in all the independent sets of the sequence the vertex $v$ by the vertex $u$. Note that all these sets still have size $k$. Indeed $(u,v)$ is an edge of $G$ and then no independent set contain both $u$ and $v$.
 Moreover after this replacement, all the sets are still independent since $N[u] \subseteq N[v]$. \qed
\end{proof}

Note moreover that if $I$ (or $J$) contains $v$ then we can transform $I$ into $I \cup \{ u \} \setminus \{ v \}$ in one step. Lemma~\ref{lem:NoInclusion} combined with this remark ensures that we can reduce the graph in such a way no vertex satisfies the second point of Lemma~\ref{lem:neighbourhoods}.

\begin{lemma}\label{lem:finvc1}
     Let $G$ be a graph of VC-dimension at most $1$ such that no pair of vertices satisfies the second point of Lemma  \ref{lem:neighbourhoods}. Let $I$ and $J$ be two independent sets of size at least $3$, then $I\cup J$ is an independent set.
\end{lemma}

\begin{proof}
      Let $u\in I$ and $v\in J$ and assume that $uv\in E$. Since there is no inclusion between the closed neighborhoods of $u$ and $v$, and since their closed neighborhoods are not disjoint, by Lemma \ref{lem:neighbourhoods}, $u,v$ forms a dominating pair. Now since $I$ is independent and $u\in I$, $v$ is adjacent to every vertices of $I$. Now let $w$ be a vertex of $I$ different from $u$. We have $v\in N[u]\cap N[w]$, so the closed neighborhoods of $u$ and $w$ are not disjoint. Thus by Lemma \ref{lem:neighbourhoods} $u,w$ is also a dominating pair. So if $I$ contains a vertex different from $u$ and $w$, it should be adjacent to at least one of them, contradicting the fact that $I$ is independent. \qed
\end{proof}

This completes the proof of Theorem \ref{thm:VC1} since either $I$ and $J$ have size at most $2$ and the problem is obviously polynomial. Or $I\cup J$ is an independent set and one can simply move every vertex  from $I$ to $J$ one by one.

\begin{question}
 Is the \tj{} problem FPT on graphs of VC-dimension at most~$2$?
\end{question}

\bibliographystyle{plain}
\bibliography{biblio}

\end{document}